\documentclass[final]{siamltex}

\usepackage{todonotes}

\usepackage{amsmath,amstext,amsbsy,amssymb,graphicx,subfig,mathdots}
\usepackage{algorithm,algpseudocode}
\usepackage{mathrsfs}
\usepackage{default,tikz}
\usetikzlibrary{positioning}
\usetikzlibrary{shapes,arrows}

\newenvironment{remark}{{\bf Remark}}

\newenvironment{mat}{\left[\begin{array}{ccccccccccccccc}}{\end{array}\right]}
\newcommand\bcm{\begin{mat}}
\newcommand\ecm{\end{mat}}

\def\O{{\cal O}}

\def\K{{\cal K}}

\def\Ai{{\text{Ai}}}

\def\I{{\rm i}}

\def\RR{{\mathbb R}}

\def\abs#1{\left|{#1}\right|}

\def\dkf{{\mathrm d}}
\def\dx{\dkf x}

\definecolor{DarkGreen}{rgb}{0,.55,0}

\definecolor{Yellow}{rgb}{.55,.55,0}

\def\Ki_#1^#2{{\mathscr K}_{#1}^{#2}}

\def\sopmatrix#1{\begin{pmatrix} #1 \end{pmatrix}}

\def\secref#1{Section~\ref{sec:#1}}

\def\figref#1{Figure~\ref{fig:#1}}

\def\addtab#1={#1\;&=}

\def\meeq#1{\def\ccr{\\\addtab}
 \begin{align*}
 \addtab#1
 \end{align*}
  }

\def\vc#1{\mbox{\boldmath$#1$\unboldmath}}

\def\Tr{{\rm Tr}\,}

\def\Figure[#1]#2\par{
\begin{figure}[tb]
\begin{center}{
\includegraphics{Figures/#1}}
\end{center}
\caption{#2}\label{fig:#1} 
\end{figure}
}
\def\Figurew[#1]<#2>#3\par{
\begin{figure}[tb]
\begin{center}{
\includegraphics[width=#2]{Figures/#1}}
\end{center}
\caption{#3}\label{fig:#1} 
\end{figure}
}

\def\Figuretwow[#1,#2]<#3>#4\par{
\begin{figure}[tb]
\begin{center}{
\includegraphics[width=#3]{Figures/#1}\includegraphics[width=#3]{Figures/#2}}
\end{center}
\caption{#4}\label{fig:#1} 
\end{figure}
}
\def\Figuretwo[#1,#2]#3\par{
	\Figuretwow[#1,#2]<0.48 \hsize>
		#3\par	
}

\def\Figurefour[#1,#2,#3,#4]#5\par{
\begin{figure}[tb]
\begin{center}{
\includegraphics[width=.48\hsize]{Figures/#1}\includegraphics[width=.48\hsize]{Figures/#2}}
\end{center}
\begin{center}{
\includegraphics[width=.48\hsize]{Figures/#3}\includegraphics[width=.48\hsize]{Figures/#4}}
\end{center}
\caption{#5}\label{fig:#1} 
\end{figure}
}

\def\hexnumber#1{\ifcase#1 0\or1\or2\or3\or4\or5\or6\or7\or8\or9\or
 A\or B\or C\or D\or E\or F\fi}

\newfam\ibfam
\font\tenib=cmmib10 \textfont\ibfam=\tenib  
\font\sevenib=cmmib7 \scriptfont\ibfam=\sevenib
\font\fiveib=cmmib5 \scriptscriptfont\ibfam=\fiveib

\edef\ibhx{\hexnumber\ibfam}

\def\D{{\rm d}}

\mathchardef\varphiB="0\ibhx27
\mathchardef\psiB="0\ibhx20
\mathchardef\kappaB="0\ibhx14
\mathchardef\sigmaB="0\ibhx1B
\mathchardef\phiB="0\ibhx1E
\mathchardef\xiB="0\ibhx18
\mathchardef\muB="0\ibhx16
\mathchardef\lambdaB="0\ibhx15

\def\pr(#1){\left({#1}\right)}
\def\br[#1]{\left[{#1}\right]}
\def\set#1{\left\{{#1}\right\}}
\def\ip<#1>{\left\langle{#1}\right\rangle}
\def\iip<#1>{\left\langle\!\langle{#1}\right\rangle\!\rangle}

\def\fpr(#1){\!\pr({#1})}

\def\function#1#2{\expandafter\def\csname #1\endcsname(##1){#2\fpr({##1})}
				\expandafter\def\csname #1p\endcsname(##1){#2'\fpr({##1})}
				\expandafter\def\csname #1pp\endcsname(##1){#2''\fpr({##1})}
				\expandafter\def\csname #1pn\endcsname##1(##2){#2^{\pr({##1})}\fpr({##2})}				}
\def\defoperator#1#2{\expandafter\def\csname #1\endcsname[##1]{#2\!\br[{##1}]}}

\def\ffunction#1{\function{#1}{#1}}

\def\O(#1){{\cal O}\!\left(#1\right)}

\def\Oo(#1){{\rm o}\!\left({#1}\right)}

\def\fO(#1){\Oh\fpr({#1})}

\def\dkfn^#1#2{\,{\rm d}^#1 #2}
\def\dkf#1{\,{\rm d}#1}

\def\I{{\rm i}}
\def\E{{\rm e}}

\def\H_#1^#2(#3){H_#1^{(#2)}\fpr({#3})}
\def\J_#1(#2){{\rm J}_#1\!\pr(#2)}

\def\dx{\dkf{x}}

\def\ddxn^#1{{\dkf{}^#1 \over \dkfn^#1{x}}}

\def\norm#1{\left\| #1 \right\|}

\def\seq_#1^#2#3{\set{#3_{#1},\ldots,#3_{#2}}}

\def\abs#1{\left|{#1}\right|}

\ffunction{f}
\ffunction{g}

\def\mapengine#1,#2.{\mapfunction{#1}\ifx\void#2\else\mapengine #2.\fi }

\def\map[#1]{\mapengine #1,\void.}

\def\mapenginesep_#1#2,#3.{\mapfunction{#2}\ifx\void#3\else#1\mapengine #3.\fi }

\def\mapsep_#1[#2]{\mapenginesep_{#1}#2,\void.}

\def\vcbr[#1]{\pr({#1})}

\def\bvect[#1,#2]{
{
\def\dots{\cdots}
\def\mapfunction##1{\ | \  ##1}
	\sopmatrix{
		 \,#1\map[#2]\,
	}
}
}

\def\vect[#1]{
{\def\dots{\ldots}
	\vcbr[{#1}]
}}

\def\vectt[#1]{
{\def\dots{\ldots}
	\vect[{#1}]^{\top}
}}

\def\Vectt[#1]{
{
\def\mapfunction##1{##1 \cr} 
\def\dots{\vdots}
	\sopmatrix{
		\map[#1]
	}
}}

\def\qfor{\quad\hbox{for}\quad}

\def\simlimit_#1{\,\,\,\sim \!\!\!\!\!\!\!\!\!{ \atop \scriptscriptstyle #1 }}

\def\XXint#1#2#3{{\setbox0=\hbox{$#1{#2#3}{\int}$}
     \vcenter{\hbox{$#2#3$}}\kern-.5\wd0}}

\def\rad^#1#2{\,\,{}^#1\!\!\!\!\sqrt{#2}\,}

\def\Figuretwofixed#1#2#3\par{
\Figuretwow{#1}{#2}{0.48 \hsize}{#3}\par
}

\newenvironment{example}{{\bf Example}}

\title{Sampling unitary invariant ensembles}

\author{
Sheehan Olver\thanks{
The
University of Sydney, Australia. (Sheehan.Olver@sydney.edu.au)}
\and
Raj Rao Nadakuditi\thanks{University of Michigan, USA. (rajnrao@umich.edu)}
\and
Thomas Trogdon\thanks{Courant Institute, New York University, USA. (trogdon@cims.nyu.edu)}
} 

\begin{document}
\maketitle

\begin{abstract}

	We develop an algorithm for sampling from the unitary invariant random matrix ensembles.  The algorithm is based on the representation of their eigenvalues as a determinantal point process whose kernel is given in terms of orthogonal polynomials.   Using this algorithm, statistics beyond those  known through analysis are calculable through Monte Carlo simulation.  Unexpected phenomena are observed in the simulations.

\end{abstract}
\section{Introduction}

The   {\it unitary invariant ensembles}  (UIE) are an important class of random matrices which are invariant under conjugation by unitary matrices.  This corresponds to the physical situation where the frame of reference does not affect the underlying statistics.    UIEs are defined as   $n \times n$ random Hermitian matrices
	$$M = \sopmatrix{
		M_{11} & M_{12}^{\rm R} + \I M_{12}^{\rm I} &\cdots & M_{1n}^{\rm R} + \I M_{1n}^{\rm I} \cr
		 M_{12}^{\rm R} - \I M_{12}^{\rm I} & M_{22 }& \cdots & M_{2n}^{\rm R} + \I M_{2n}^{\rm I} \cr
		 \vdots & \ddots & \ddots & \vdots \cr
		 M_{1n}^{\rm R} - \I M_{1n}^{\rm I}  & \cdots &		 M_{(n-1)n}^{\rm R} - \I M_{(n-1)n}^{\rm I} & M_{nn}
	}$$
whose entries are distributed according to a given {\it potential} $Q$, by the rule
	$${1 \over Z_n} \E^{- \Tr Q(M)} \dkf M,$$
where $Z_n$ is the normalization constant and
	$$\dkf M = \prod_{i = 1}^n \dkf M_{ii} \prod_{i < j} (\dkf M_{ij}^{\rm R} \dkf M_{ij}^{\rm I}).$$

Associated with the UIE potential $Q$ is the weight $w(x) = \E^{- Q(x)}$.  We include  weights of the form $x^\alpha \E^{- Q(x)}$, which corresponds to the UIE

	$$\E^{- \Tr \br[ Q(M) -\alpha \log M ]} =  (\E^{ \Tr \log M })^\alpha  \E^{- \Tr Q(M)} = (\det M)^\alpha \E^{- \Tr Q(M)}, $$
using $\det \E^M = \E^{\Tr M}$.  
We also allow $Q$ to depend on $n$, particularly $Q(x) = n V(x)$.

%

	The contribution of this paper is the development of an efficient algorithm for sampling from invariant ensembles.  Sampling the entries directly is prohibitively expensive: it is difficult to sample many random variables that depend on each other in a complicated manner.    Instead we exploit the fact that the eigenvalues of invariant ensembles are described by determinantal point processes whose kernel is written in terms of the associated orthogonal polynomials.  Thus, the task of sampling invariant ensembles is reduced to the following:
\begin{enumerate}
	\item Construct the orthogonal polynomials associated to the weight $w$.
	\item Sample a determinantal point process defined through this sequence of orthogonal polynomials.  
\end{enumerate}
The first task can be accomplished via either Stieljes procedure \cite{GautschiOP}, or using Riemann--Hilbert techniques \cite{TrogdonSOGauss}. For the second task, we adapt a recently developed algorithm \cite{HoughDetSamp,ScardicchioDetSamp}.

	Very recently, an alternative approach for sampling invariant ensembles based on simulating Dyson Brownian motion was developed by Li and Menon \cite{LiMenonDetSamp}.  Both approaches have the same complexity of ${\cal O}(n^3)$ operations to sample an $n \times n$ matrix,  although the method of Li and Melon is likely to be significantly faster.  However, our approach samples the correct distribution to essentially machine precision accuracy, a task that is computationally impractical using numerical simulation of stochastic differential equations.  See Section~\ref{sec:compare} and \figref{KolmogorovSmirnov} for a demonstration and discussion of this.

\bigskip

	Our motivation for sampling invariant ensembles is to extend the class of matrices for which simulations can be performed, to facilitate a deeper understanding of the rich phenomena that random matrices exhibit.  Currently, the ensembles that can easily be sampled are restricted to those that are generated from independent entries.  These include the following classical ensembles:
\begin{itemize}
	 \item  {\it Wigner ensembles}:  Hermitian ensembles with independent entries.
	 \item  {\it Wishart ensembles}:  $X X^\star$ where $X$ is rectangular with independent complex entries.
	 \item {\it Manova ensembles}: $A (A+B)^{-1}$ for $A = a^\star a$, $B = b^\star b$ and $a,b$ rectangular with independent complex entries.     
\end{itemize}	 
		We note that three of the most studied ensembles lie right at the intersection of these ensembles and invariant ensembles:
\begin{itemize}
	\item {\it Gaussian Unitary Ensemble (GUE)}: Wigner ensemble whose  entries are independent complex Gaussian with twice the variance on the diagonal, is also a UIE with  potential $Q(x) = x^2$.
	\item {\it Laguerre Unitary Ensemble (LUE)}: Wishart ensemble where $X$ is rectangular with iid complex Gaussian entries, is also an UIE with weight $x^\alpha \E^{-x}$ on the half line.
	\item {\it Jacobi Unitary Ensemble (JUE)}: Manova ensemble  where $a,b$ are rectangular with iid complex Gaussian entries, is also a UIE with weight $(1-x)^\alpha (1+x)^\beta$ \cite[pp.~1440]{WitteJUEGap}.
\end{itemize}

While  ensembles generated from independent entries are easily sampled --- particularly GUE, LUE and JUE which have banded representations that enable faster computation \cite{DumitriuBetaEnsemble,EdelmanJUECSDecomposition} --- they are  by no means general: under broad conditions \cite{WignerSemicircle,PasturSemicircle},   Wigner ensembles follow the semicircle law --- i.e., the global distribution of eigenvalues tends to a semicircle --- while  GUE is the only invariant ensemble that follows the semicircle law.  Similarly, Wishart ensembles follow the Mar{\v{c}}enko--Pastur law \cite{MarchenkoPastur}.   On the other hand, whereas simulation has been primarily limited to  ensembles generated from independent entries, analysis of invariant ensembles is quite developed (cf.~\cite{DeiftOrthogonalPolynomials,DeiftInvariantEnsembles}).

A particular phenomena studied in depth is {\it universality}, and universality laws for many canonical families of random matrices, in partigular Wigner, Wishart and invariant ensembles,  have been proved in the bulk \cite{DeiftOrthogonalPolynomials,LubinskyBulkUniversality,YauBulkUniversality} and edge \cite{TaoEdgeUniversality,DeiftInvariantEnsembles}.   The techniques of these proofs differ greatly, hinting that universality is applicable to a much wider class of matrices, beyond the reach of current analytical approaches.
Using the results of this paper, it is now possible to guide this conjecture  through simulation.  For example, we include in \secref{Experiments} some experiments that suggest algebraic manipulations of invariant ensembles also follow universality laws.

	Notably, a very special subset of invariant ensembles have been shown to violate the Tracy--Widom universality law; rather, they follow higher-order analogues of the Tracy--Widom law \cite{HigherOrderTracyWidom}.  While this was discovered analytically, using the sampling algorithm we are able to verify it through simulation.   We further show that adding a sufficient number of these degenerate ensembles (precisely, four) causes the standard Tracy--Widom law to reemerge.  This   demonstrates why broadening the possible simulations is important: if one is restricted to ensembles  generated from independent entries, such phenomena may have gone undiscovered.  

\bigskip

The paper is organized as follows:

\begin{description}
	\item{\secref{IEs}:}  We give a brief overview of invariant ensembles and the distribution of their eigenvalues.
	\item{\secref{Experiments}:}  We present experimental results that are achieved using the sampling algorithm,  verifying that special classes of ensembles do indeed follow alternative universality results.   We also compare the sampling accuracy of our approach to that of Li and Menon \cite{LiMenonDetSamp}.
	\item{\secref{OPs}:}  We give background on orthogonal polynomials and recurrence relationships, and describe their approximation by Chebyshev series for use in the algorithm.
	\item{\secref{DetProc}:}  We close our discussion with a brief overview of determinantal point processes, and present the algorithm as specialized for the orthogonal polynomial setting.  We prove that the algorithm does indeed sample unitary invariant ensembles.  
\end{description}

%

\section{Unitary Invariant Ensembles}\label{sec:IEs}

We begin with a brief overview of how the eigenvalues of a UIE are reduced to a determinantal point process.    The eigenvectors of any UIE are simply Haar distributed unitary matrices\footnote{This follows from GUE having Haar distributed eigenvectors \cite[Corollary 2.5.4]{AndersonRMT} and the distribution of the eigenvectors of all UIEs being independent of $V$ \cite[pp.~25]{DeiftInvariantEnsembles}.}.  Upon integrating out the eigenvectors, the 
%
%
	 distribution of the eigenvalues is determined to be \cite[Section 5.4]{DeiftOrthogonalPolynomials}
	\begin{equation}\label{eq:evdist}
		{1 \over \hat Z_n} \prod_{i <j} (\lambda_i - \lambda_j)^2  w(\lambda_1) \cdots w(\lambda_n) \D \lambdaB
	\end{equation}
where $\hat Z_n$ is a normalization constant.

This distribution can be rewritten in terms of the determinant of orthogonal polynomials.  Let $p_0,p_1,\ldots$ be the polynomials orthonormal with respect to
	$w(x) \D x,$
and define
	 $$\phi_k(x) := p_k(x) \sqrt{w(x)},$$
which are orthonormal in ${\rm L}^2(\mathbb R)$.    Define  the kernel
	\begin{equation}\label{eq:kernel}
		{\cal K}_n(x,y) = \sum_{k=0}^{n-1} \phi_k(x)\phi_k(y).
	\end{equation}
then we have \cite[(5.30)]{DeiftOrthogonalPolynomials} 
	$$\hbox{\eqref{eq:evdist}} = {1 \over n!} \det \begin{pmatrix}
	{\cal K}_n(\lambda_1,\lambda_1) & \cdots & 	{\cal K}_n(\lambda_n,\lambda_1) \cr
	\vdots & \ddots & \vdots \cr
	{\cal K}_n(\lambda_1,\lambda_n) & \cdots & 	{\cal K}_n(\lambda_n,\lambda_n) 
	\end{pmatrix} \D\lambdaB.$$
This form of distribution defines a {\it determinantal point process} (see, e.g., \cite[Chapter 11]{RMTHandbook}).

\subsection{Spectral densities and equilibrium measures}

	The spectral density of a random matrix is the distribution of any single eigenvalue, without ordering.  In other words, we integrate out all but one of the eigenvalues in \eqref{eq:evdist}.  The so-called integrating out lemma \cite[pg.~103]{DeiftOrthogonalPolynomials} tells us that the spectral density is precisely
%
	$${\K_n(x,x) \over n} \D x.$$

As $n \rightarrow \infty$, the spectral density approaches the {\it limiting spectral density}.  In the case of UIE with $Q(x) = n V(x)$, the scaling by $n$  causes the limiting spectral density to have compact support. 	The limiting spectral density is then precisely the equilibrium measure of $V$ \cite[Section 6.4]{DeiftOrthogonalPolynomials}:			
				
				\begin{definition}\label{def:EM}
	The {\it equilibrium measure} $\mu$ is the unique minimizer of
	$$\iint\log{ \frac{1}{|x - y|}} \dkf\mu(x) \dkf\mu(y) +  \int  V(x) \dkf\mu(x)$$
among Borel probability measures on $\mathbb R$.
\end{definition}

		
	Some ensembles of interest do not have the scaled form $n V(x)$.  From classical probability, however, we see that multiplication of the random matrix by a constant $\alpha(n)$ induces a scaling of $Q$: if $M$ is an invariant ensemble with potential $Q(x)$, then $\alpha(n)^{-1} M$ is also an invariant ensemble with potential $Q(\alpha(n) x)$.  We can chose $\alpha(n)$ so that $Q(\alpha(n) x)$ has the desired form.

\begin{example}
	The unscaled GUE has potential $Q(x) = x^2$.  Multiplication by ${1 \over \sqrt n}$ gives the new potential $Q(\sqrt{n} x) = n x^2 = n V(x)$, which is of the desired form.
\end{example}

\begin{example}
	Consider a general polynomial potential $Q(x) = q_m x^m + \cdots q_0$.  Multiplying the ensemble by ${1 \over  n^{1/m}}$ gives the new potential
	 $$Q(n^{1/m} x) = n (q_m x^m + n^{-{1 \over m}} q_{m-1} x^{m-1} + \cdots + n^{-1} q_0)   = n V_n(x).$$
 While $V_n$ now depends on $n$, it tends to the $n$ independent monomial $q_m x^m$, and the limiting spectral density is the equilibrium measure\footnote{This follows from the leading order asymptotics of the equilibrium measure of $V_n$ only depending on the leading order of the polynomial \cite[(5.26)]{DeiftWeights1}, and the expression of the finite spectral density as ${\K_n(x,x) \over n} \D x$.} of $q_m x^m$.    If we can sample $M$ from the invariant ensemble with potential $Q(x) = n V_n(x)$, then $n^{1/m} M$ gives a sample from the unscaled $Q(x)$.  
\end{example}

\begin{remark}
 If we restrict our attention to weights supported on compact sets (say, $[-1,1]$), then it is not necessary to induce a scaling.  When $Q$ is entire but not polynomial, there is no explicit scaling, however, a choice of $\alpha(n)$ that causes the equilibrium measure to have compact support for all $n$  can be computed numerically \cite{TrogdonSOGauss}.  
\end{remark}

For the potentials $V$ we consider, the equilibrium measure is supported on a single interval $[a_V,b_V]$.  The measure  can be readily calculated using \cite{SOEquilibriumMeasure}, returning an approximation to the representation in terms of Chebyshev U series:
	$$\dkf \mu(x)  = \psi(x) \dx \qfor  \psi\fpr({b_V+a_V \over 2} + {b_V -a_V \over 2} x) =   {\sqrt{1-x^2} \over 2 \pi} \sum_{k=1}^\infty V_k U_{k-1}(x) \dx,$$
where $V_k$ are the Chebyshev coefficients of $V'$:
	$$V'\fpr({b_V+a_V \over 2} + {b_V -a_V \over 2} x) = \sum_{k=0}^\infty V_k T_k(x).$$

Associated with the equilibrium measure is a scaling constant for the edge universality law.    In the non-degenerate case where the equilibrium measure has a precisely square root singularity at $b_V$, this is \cite{SOTrogdonRMT}
	\begin{equation}\label{eq:edgescalingconstant}
		c_{V} = (b_V - a_V)^{-1/3} \pr( 2 \pi \sum_{k=1}^\infty V_k)^{2/3}.
	\end{equation}

\subsection{Gap statistics and universality}

	

	The gap statistics are local statistics, measuring the probability  that there are no eigenvalues in an interval $\Omega$.  These statistics can be written in terms of a Fredholm determinant \cite[Proposition 4.6.2]{RMTHandbook}: 
	$$\det (I - \K_n |_{L^2(\Omega)}).$$
 The notation  $\K_n |_{L^2(\Omega)}$ refers to the integral operator on $L^2(\Omega)$ with kernel $\K_n$.  A special case of gap statistics is the edge statistic, which is the probability that the largest eigenvalue is greater than $s$; i.e., $\Omega = [s,\infty)$.

\subsection{Universality}\label{sec:Universality}

There are two types of universality laws we discuss: {\it bulk universality} and {\it edge universality}.  The bulk universality law states that the gap statistic of a scaled neighbourhood of $x$ in the support of the limiting spectral measure of a Hermitian random matrix tends (in a appropriate sense) to the sine-kernel law: 
	$$ \det (I - \K_n |_{L^2\br[x + {(-s,s) \over \psi(x) n}]})\rightarrow \rho^{\sin}(s),$$
where $\psi$ is the density of the equilibrium measure and 
	$$\rho^{\sin}(s) = \det(I - {\cal S}|_{L^2(-s,s)})\qfor {\cal S} = {\sin(x-y) \over x - y}.$$
This universality law has been proved under broad conditions for both Wigner ensembles \cite{TaoBulkUniversality,YauBulkUniversality} and UIE \cite{DeiftOrthogonalPolynomials,LubinskyBulkUniversality}.  

Soft edge universality states that the distribution of a scaled largest eigenvalue tends to the Tracy--Widom law:
	$$\det \pr(I - \K_n |_{L^2\br[(b_V +  {s \over c_V n^{2/3}},\infty)]})\rightarrow \rho^{\Ai}(s)$$
 where $b_V$ is again the right endpoint of the equilibrium measure, $c_V$ is the edge scaling constant \eqref{eq:edgescalingconstant} and
 $$  \rho^{\Ai}(s) = \det(I - {\cal A}|_{L^2(s,\infty)})\qfor {\cal A}(x,y) = {\Ai(x) \Ai'(y) - \Ai'(x) \Ai(y) \over x - y}.$$
%
%
 This law only applies if the behaviour at the right endpoint of the limiting spectral density has precisely square root decay.  Since Wigner ensembles satisfy the semicircle law (having precisely square root decay), they broadly satisfy the Tracy--Widom law \cite{TaoEdgeUniversality}.    When the limiting spectral density has a singularity, as can be the case at the left endpoint of  a Wishart ensemble,  hard edge universality applies, which is described in terms of the Bessel kernel \cite[pp.~107]{RMTHandbook}.    

 However, invariant ensembles can have limiting spectral densities that  exhibit faster decay, in which case the largest eigenvalue follows the higher-order Tracy--Widom law  \cite{ClaeysVanlessenHigherOrderTW,HigherOrderTracyWidom}.  Below, \secref{higherorder}, is the first time the higher-order Tracy--Widom law have been observed through Monte Carlo simulations. 
 
	

\section{Experimental results}\label{sec:Experiments}

In this section, we present preliminary experimental results that are now possible using the algorithm for sampling invariant ensembles.  


%
%
%
%
%


\subsection{Quartic ensemble}\label{sec:sampleQuartic}

The quartic ensemble has the potential $V(x) = x^4$.   In the left-hand side of \figref{QuarticDecayEdge}, we see the convergence of the empirical CDF of the shifted and scaled largest eigenvalue
	$c_V (\lambda_{\rm max} - b_V)$
 to the Tracy--Widom law.  In \figref{QuarticBulkConvergence} we plot the difference between the empirical CDF and the Tracy--Widom law, and also compare the empirical  complementary CDF\footnote{The empirical complementary CDF is the probability that a random variable is larger than a value $x$, as measured through Monte Carlo simulation.}  of 
 	$$n \psi(0) \abs{\lambda_{\rm min}}$$
  to the sine kernel law, where $\lambda_{\rm min}$ is the eigenvalue with smallest absolute value and  $\psi$ is the density of the equilibrium measure.

\begin{remark}
	Monte Carlo simulation for the bulk statistics in \figref{QuarticBulkConvergence} and \figref{CoshBulkConvergence} had not fully converged to the true distribution after a million samples, hence these should be taken as a rough estimate.  For invariant ensembles, high accuracy and fast calculation of the true statistic is possible via the methodology of \cite{SOTrogdonRMT}.  
\end{remark}

\Figuretwo[QuarticDecayEdge,HODecayEdge]
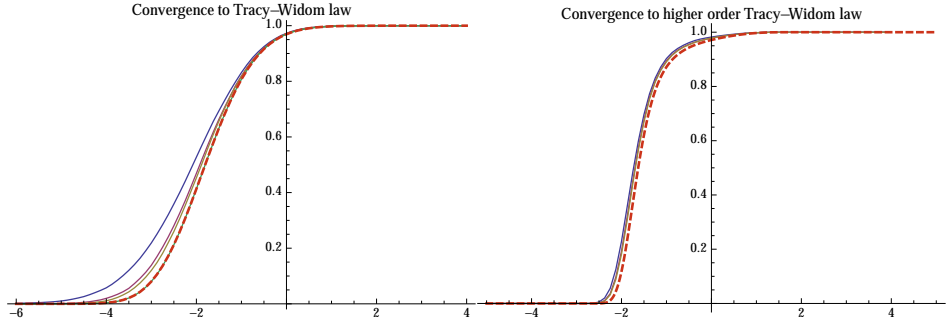
	Monte Carlo simulation with $n = 5, 15$ and 25 of the quartic ensemble (left) and the higher-order decay ensemble (right) compared to the standard and higher-order Tracy--Widom distributions, respectively (dashed line).

\Figuretwo[QuarticBulkConvergence,QuarticEdgeConvergence]
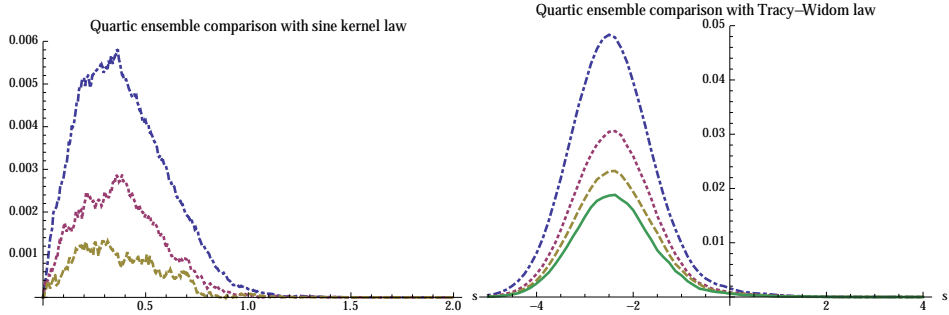
	Convergence of Monte Carlo simulation to universality for $w(x) = \E^{-x^4}$, for $n = 25$ (dot-dashed), 50 (dotted), 75 (dashed) and 100 (plain), with a million samples.  Left: difference between the empirical bulk statistic and the sine kernel law ($n = 100$ not pictured as the distribution was not sufficiently resolved).  Right: convergence of the edge statistic to the Tracy--Widom law.

%
%
%

\subsection{Higher-order decay ensemble}\label{sec:higherorder}

A special ensemble has the potential
	$$V(x) = {x^4 \over 20}   - {4 \over 15} x^3 + {x^2 \over 5} + {8 \over 5} x,$$
and the equilibrium measure \cite{HigherOrderTracyWidom}
	$$ \psi(x)\dkf x = {1 \over 10 \pi} (x + 2)^{1/2} (2 - x)^{5/2} \dkf x \qfor -2 < x < 2,$$
which is the dashed line in the left-hand side of \figref{HODecayAddEM}.  The higher-order decay of the equilibrium causes the large eigenvalue statistics to be different: rather than following the Tracy--Widom distribution, they follow the higher-order Tracy--Widom distributions \cite{HigherOrderTracyWidom}.

	In \figref{QuarticDecayEdge}, we compare Monte Carlo simulations with the finite $n$ distributions to the predicted higher-order Tracy--Widom distribution as calculated in \cite{SOClaeysHigherOrderTW}.    We plot the edge statistic in the neighbourhood $[2 + {s \over c_V^{\rm HO} n^{2/7}},\infty)$,  for the scaling constant\footnote{The constant is stated in \cite{ClaeysVanlessenHigherOrderTW} as $6^{1/7}$.  However, the normalization used in \cite{HigherOrderTracyWidom,SOClaeysHigherOrderTW} differ by a factor of 30, giving $  c_{V}^{\rm HO} = (6/30)^{1/7} = 5^{-2/7}$.} $c_{V}^{\rm HO} = 5^{-2/7}$.

	\begin{remark}
		While the Monte Carlo simulations roughly follow the predicted universality law, they differ by a substantial amount due to $n$ not being sufficiently high.  We have verified the accuracy of Monte Carlo simulation by comparing to a high accuracy approximation of the finite $n$ statistic   using a numerical Riemann--Hilbert approach  \cite{SOTrogdonRMT} (not pictured).  Using the Riemann--Hilbert approach, we see that high accuracy agreement of the finite $n$ statistic with the universality law does not occur until $n \approx 100,000$,  beyond the current reach of simulation.  
		\end{remark}

\subsection{Non-varying cosh ensemble}

	Consider the weight
	$$w(x) = \E^{-\cosh x}$$
which does not scale with $n$, hence the existing theory on universality breaks down.    However, we can still perform sampling as we can calculate the recurrence relationship for the associated orthogonal polynomials, using the Riemann--Hilbert approach \cite{TrogdonSOGauss}.   

	With sampling in hand, we verify that the ensemble does indeed follow both bulk and edge universality in \figref{CoshBulkConvergence}.  Here, to get convergence to the universality laws,  we choose the scaling constants associated with the equilibrium measure of ${1 \over n} \cosh x$.  It is important to note that in  this case $c_V$ varies with $n$.  The convergence to the Tracy--Widom law of the edge statistic is surprising: the equilibrium measure should not have a nice limit at the right endpoint.  

\Figuretwo[CoshBulkConvergence,CoshEdgeConvergence]
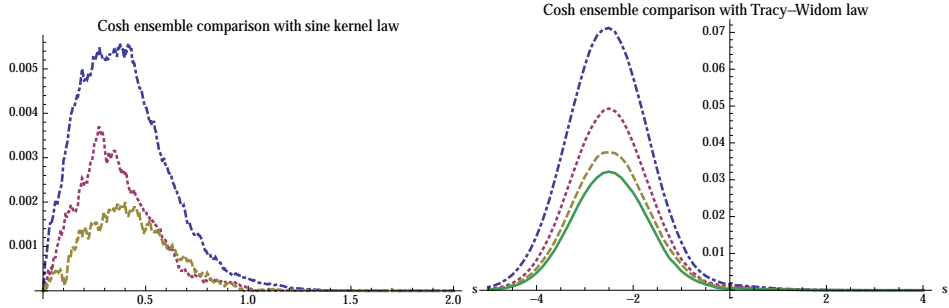
	Convergence of Monte Carlo simulation to universality for $w(x) = \E^{-\cosh x}$, for $n = 25$ (dot-dashed), 50 (dotted), 75 (dashed) and 100 (plain).  Left: convergence of the bulk statistic to the sine kernel law.  Right: convergence of the edge statistic to the Tracy--Widom law.

\subsection{Addition of invariant ensembles}

	While the statistics of invariant ensembles themselves are already known in detail, both asymptotically \cite{DeiftOrthogonalPolynomials} and numerically \cite{SOTrogdonRMT}, the theory and numerics break down as soon as further manipulations (\emph{e.g.} matrix addition) of invariant ensembles are performed.  The ability to sample invariant ensembles directly means that we can now perform Monte Carlo simulations to help understand such manipulations.  In this section, we investigate the global distribution of the eigenvalues for the addition of invariant ensembles.  
	
	Our first experiment consists of calculating the spectral density of quartic ensemble added to a GUE.  The $n = 1$ is equivalent to classical convolution, while as $n$ increases it approaches the free probability convolution, which has been calculated symbolically in \cite{RaoEdelmanPolynomialMethod} and numerically \cite{OlverRaoNumericalFreeProb}.    In \figref{QuarticGUE1}, we see that the simulated statistics agree in the $n=1$ case, as expected.  Surprisingly, close agreement is seen for $n = 10$ between simulation and the free probability convolution.  This relationship is only known to hold in the large $n$ limit.   
	

\Figurefour[QuarticGUE1,QuarticGUE2,QuarticGUE5,QuarticGUE10]
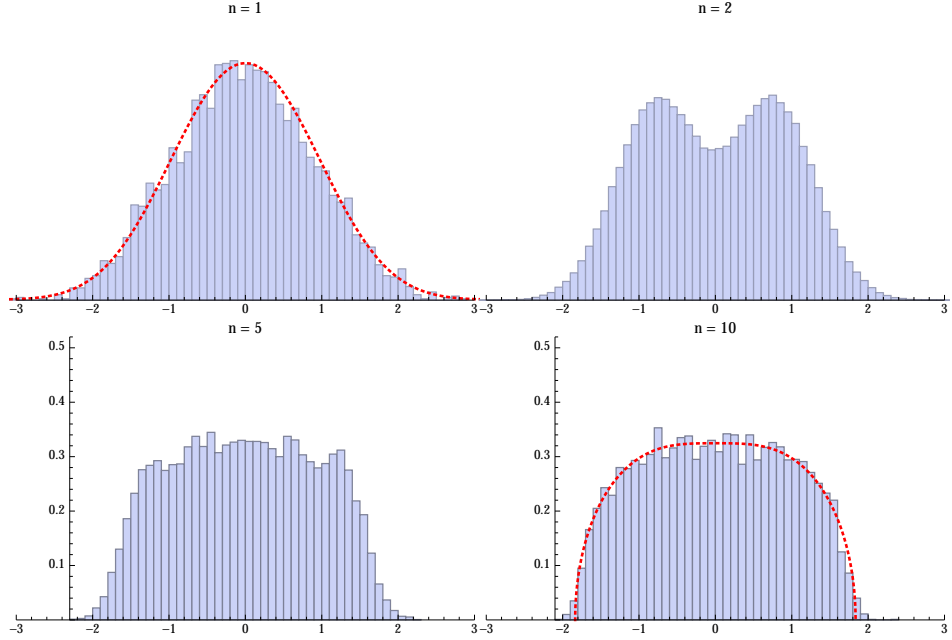
	Spectral density of the addition of an $n \times n$ Quartic ensemble and GUE for $n = 1, 2, 5$ and 10.  The dotted line in the first figure is a numerical classical additive convolution, and the dotted line in the last figure is the numerically calculated free additive convolution.  Close agreement is seen for $n = 10$ between simulation and the free probability convolution ($n = \infty$).


 In the second example, we sample $H_1,\ldots,H_k$ independent higher-order decay ensembles, as defined in \secref{higherorder}, and investigate the statistics of
 	$H_1 + \cdots  + H_k$.
Free probability tells us the limiting spectral density, which we show in the left-hand side of \figref{HODecayAddEM},  calculated numerically via \cite{OlverRaoNumericalFreeProb} and verified via Monte Carlo simulation (not pictured).  On the right, we plot the empirical edge statistic: the CDF of   $ .3 n^{2/3} (\lambda_{\rm max} - b_k)$, where $b_k$ is the right endpoint of the limiting spectral density.  (The choice of  scaling by $.3$ is arbitrary.)  In the limiting spectral density we see the emergence of a precisely square root singularity at $k = 4$.  This coincides with the edge statistic appearing to follow the standard Tracy--Widom law, rather than the higher-order analogues.

 \Figuretwo[HODecayAddEM,HODecayAddEdge]
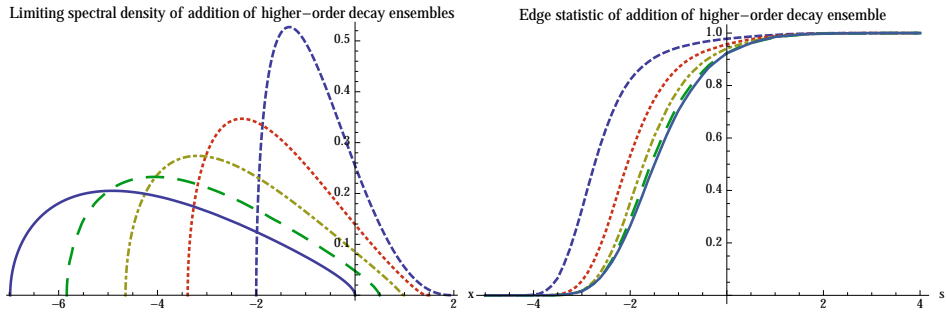
 	Statistics of the addition of $k$ higher-order decay ensembles, for $k=1$ (dashed), 2 (dotted), 3 (dot-dashed), 4 (long-dashed) and 5 (plain).  Left: the limiting spectral densities.  Right: Monte Carlo simulation of the edge statistics.

\subsection{The Kolmogorov--Smirnov statistic}\label{sec:compare}

In following the error analysis of \cite{LiMenonDetSamp}, we consider Kolmogorov--Smirnov (KS) statistics.  The empirical distribution function for $m$ samples $\{\lambda^m\}$, each consisting of $n$ eigenvalues, is
\begin{align*}
F_{n,m}(x) = \frac{1}{nm} \sum_{j=1}^m\sum_{k=1}^n \mathbf 1_{\lambda^m_k < x}.
\end{align*}
This should be compared with both the spectral distribution function
\begin{align*}
F_n(x) = \frac{1}{n}\int_{-\infty}^x \mathcal K_n(y,y) dy, 
\end{align*}
and the limiting distribution function $F(x) = \mu((-\infty,x])$ where $\mu$ is the equilibrium measure.  The two relevant KS statistics are
\begin{align*}
E_{n,m} = \sup_{x\in \mathbb R} | F_{n,m}(x) - F_n(x)|, \quad E^\infty_{n,m} = \sup_{x \in \mathbb R}| F_{n,m}(x) - F(x)|.
\end{align*}
We show these statistics for the Quartic ensemble in \figref{KolmogorovSmirnov}.  It should be noted that the data point for $E^\infty_{100,1000}$ (see the $+$ at $n = 100$ in \figref{KolmogorovSmirnov}) appears to lie below the comparable statistic in \cite[Figure 5(b)]{LiMenonDetSamp}.  Furthermore, our errors do not saturate in the same way.  This can be explained by the fact that while truncation errors are a concern for the numerical solution of stochastic differential equations, they are much less significant (at least to $n = 1000$, $m = 5000$) for the method presented here.

\Figure[KolmogorovSmirnov]
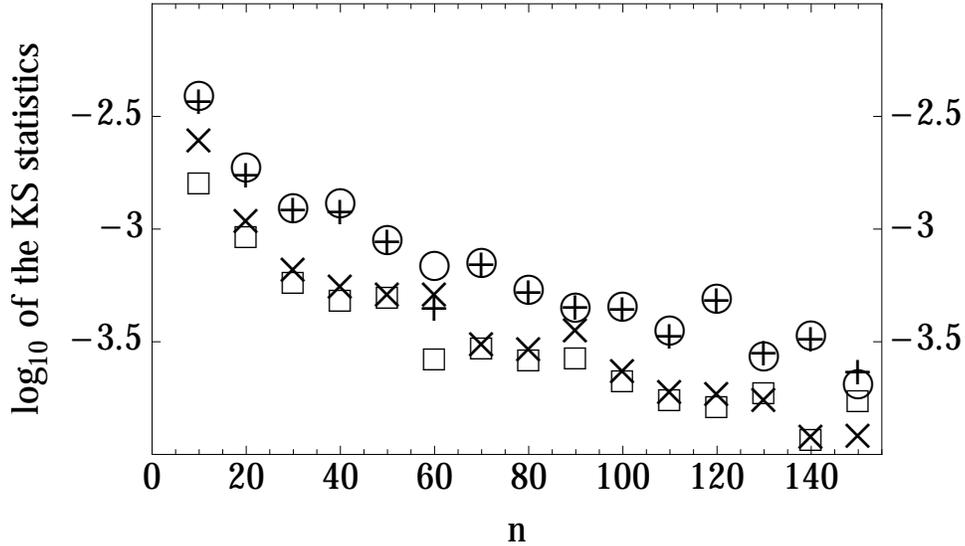
     The estimated KS statistics for the Quartic ensemble. $+$: $E^\infty_{n,1000}$ plotted versus $n$, $\bigcirc$: $E_{n,1000}$ plotted versus $n$, $\times$: $E^\infty_{n,5000}$ plotted versus $n$, $\square$: $E_{n,5000}$ plotted verus $n$.  We see an essentially monotonic decay in the statistic as the matrix size, $n$, increases.  As expected, lower errors are seen when comparing with $\mathcal K_n$.

 \section{Chebyshev expansion of weighted orthogonal polynomials}\label{sec:OPs}
 
	The algorithm below will depend on calculating $\K_n(x,y)$, which requires calculating the orthogonal polynomials.  Recall that if $p_k(x)$ are orthonormal with respect to $w(x) \dx$, then they satisfy a symmetric three-term recurrence relationship
	$$\beta_{k-1} p_{k-1}(x) + \alpha_k p_k(x) + \beta_k p_{k+1}(x) = x p_k(x)$$
with recurrence coefficients $\alpha_k$ and $\beta_k$.  We use these recurrence coefficients, along with the constant $\int w(x) \dx$,  to calculate orthogonal polynomials pointwise.

Computation of the recurrence coefficients $\alpha_k$ and $\beta_k$ is required.  These are known in closed form for classical orthogonal polynomials: Hermite, Jacobi and Laguerre polynomials.  Otherwise, the standard approach is the Stieltjes procedure \cite{GautschiOP}, which is essentially the modified Gram--Schmidt method applied to a discretized inner product using a quadrature rule.  Alternatively, one can calculate these via the numerical solution of Riemann--Hilbert problems \cite{TrogdonSOGauss}.  While the latter approach is more complicated, it is an ${\cal O}(n)$ algorithm, which is significantly better complexity than Stieljes procedure for large $n$, see the discussion in \cite{TrogdonSOGauss}.  

	Pointwise evaluation of the orthogonal polynomials is not sufficient on its own: in the algorithm, we need to perform the following further operations on $\vc\phi(x) = \vectt[\phi_0(x),\ldots,\phi_{n-1}(x)]$ where $\phi_k(x) = \sqrt{w(x)} p_k(x)$ :
\begin{enumerate}
	\item Evaluation of $\vc\phi(r)$.
	\item Pointwise evaluation of the indefinite integral of $\vc\phi(x) A \vc\phi(x)$ where $A $ is a constant matrix.
\end{enumerate}
A convenient method of accomplishing these tasks is to initially expand $\phi_k$ into a Chebyshev expansion on an interval $[a,b]$ chosen so that $\phi_k$ is negligible off the interval:
	$$\phi_k\fpr({a + b \over 2} + {b - a \over 2} x) \approx \sum_{j=0}^{m-1} c_{kj} T_j(x).$$
We determine this expansion by adaptively doubling $\tilde m$ until the last eight coefficients of
	$$\phi_0k\fpr({a + b \over 2} + {b - a \over 2} x)  \approx \sum_{j=0}^{\tilde m - 1} c_{0j} T_k(x)$$
are negligible, {\it a l\'a} the {\sc Chebfun} package \cite{chebfun}.  At each test value of $\tilde m$,  $c_{kj}$ are determined by applying the DCT to $ \phi_0\fpr({a + b \over 2} + {b - a \over 2}  \vc x_m)$ for  $m$ Chebyshev points of the first kind:
	$$\vc x_m = \vectt[-1,\cos \pi {\br[{1 - {1 \over m - 1}}]},\dots,\cos \pi {{{1 \over m - 1}}},1].$$
The negligibility of the last coefficients of $\phi_0$ means that it is effectively a polynomial of degree less than $\tilde m$, hence $\phi_1,\ldots,\phi_{n-1}$ are effectively polynomials of degree less than $\tilde m + n $, and $\phi_k(x) \phi_j(x)$ are effectively polynomials of degree less than $m = 2 (\tilde m + n)$.   From the values of $\phi_0(\vc x_m)$, we can successfully calculate the values of $\phi_k(\vc x_m)$ using the recurrence relationship.  

Using the values of $\phi_k$ at the Chebyshev points $\vc x_m$,  $\vc\phi(r)$ can be approximated in ${\cal O}(m)$ operations using the barycentric formula \cite{Bary}.   We can efficiently multiply Chebyshev series at each Chebyshev point, and then the Chebyshev expansion of $\vc\phi(x) A \vc\phi(x)$ can be calculated using the DCT.   Indefinite integration of the resulting Chebyshev expansion is possible in ${\cal O}(m)$ operations \cite[pp.~32--33]{mason2002chebyshev}, and then the resulting Chebyshev expansion  can be evaluated pointwise in ${\cal O}(m)$ operations using Clenshaw's method \cite{ClenshawsMethod}.  



\section{Sampling determinantal point processes}\label{sec:DetProc}

Here we present an algorithm for sampling the determinantal point process associated with the eigenvalues of UIEs.  Determinantal point processes  are point processes whose distribution has  a determinantal representation
	$$\det\br[{\cal K}_n(x_i,x_j)]_{i,j,=1}^n$$
 in terms of a kernel ${\cal K}_n(x,y)$, see \cite[Chapter 11]{RMTHandbook} for a general exposition.  In our case, the kernel $\K_n$ is given in terms of orthogonal polynomials, recall definition \eqref{eq:kernel}.

In  \cite{HoughDetSamp,ScardicchioDetSamp} an approach for sampling determinantal point processes was introduced.  In our concrete setting, the algorithm takes the following form.  

\begin{definition}
For $A \in \RR^{m \times n}$ with row rank $r$ and $m < n$, the function ${\rm null}(A)$ is defined to return an orthogonal matrix $Q \in \RR^{n \times n - r}$ whose columns span the kernel of $A$, so that $A Q = 0$. It will be clear that our results do not depend on the choice of matrix $Q$.
\end{definition}

\begin{algorithm}[H]
	\caption{Sample determinantal processes, adapted from \cite{HoughDetSamp,ScardicchioDetSamp} \label{ChebyshevSample}}
	\begin{algorithmic}
	\State{Input: Chebyshev interpolations of  $\vc\phi = \vectt[\phi_0,\ldots,\phi_{n-1}]$}
	\State{Output: $n$ UIE eigenvalues $\vc r = \vect[r_1,\ldots, r_n]$}
	\State{Initialize $\vc q_n(x) = \vc\phi(x)$}
	\For{ $k = n,\ldots,2$}
	\State{Obtain $r_k$ by sampling the PDF ${\vc q_k(x)^\top \vc q_k(x) \over k}$}
	\State{Let $\vc f_k = \vc q_k(r_k) \in {\mathbb R}^k$}
	\State{Let $Q_k = {\rm null}( \vc f_k^\top)$ (so that $Q_k^\top \vc f_k = 0$, and $Q_k \in \RR^{k\times k - 1}$})
	\State{Let $\vc q_{k-1}(x) = Q_k^\top \vc q_k(x)$ }
	\EndFor
	\State{Obtain $r_1$ by sampling the PDF ${\vc q_1^\top \vc q_1}$}
\end{algorithmic}
\end{algorithm}

We now use this algorithm, combined with calculation of orthogonal polynomials, to sample unitary invariant ensembles.

\begin{algorithm}[H]
	\caption{Sample unitary invariant ensembles}	
	\begin{algorithmic}
	\State{Input: Chebyshev interpolations of  $\vc\phi = \vectt[\phi_0,\ldots,\phi_{n-1}]$}
	\State{Output: $n \times n$ UIE matrix}
	\State{Obtain	$n$ eigenvalues $\vc r$ by calling Algorithm~\ref{ChebyshevSample}}
	\State{Sample unitary matrix $V$ from the Haar distribution.  E.g., orthogonalize a random $n \times n$ complex matrix whose entries are  iid complex Gaussians \cite{Stewart}.}
	\State{Return $V \diag\!\set{\vc r} V^\star$}
	\end{algorithmic}
\end{algorithm}

We must ensure that we are sampling the correct distribution with this algorithm.  Algorithm~\ref{ChebyshevSample} can be interpreted as a construction of mathematical random variables, however, on a computer these are in fact a deterministic sequence of pseudo-random numbers.  In particular, we sample ${\vc q_k^\top(x) \vc q_k(x) \over k}$ using the numerical inverse transform sampling approach of \cite{SOTownsendInverseSampling}: calculate the CDF 
  	$$F(x) = \int_{a}^x {\vc q_k(x)^\top \vc q_k(x) \over k} \dx$$ 
	using indefinite integration of the Chebyshev series representation, sample pseudo-random variable $Y$ uniformly in $[0,1]$, and find $X$ satisfying $F(X) = Y$    using the bisection method, which only requires evaluation of $F$.

We use the following definition  in order to be precise on what is meant by sampling a distribution.  Note that this definition only encodes convergence to the correct distribution; pseudo-randomness is a subtle issue beyond the scope of this paper.

\begin{definition}
A method is said to sample a Borel probability measure $\mu$ on $\mathbb R^n$ if it produces a (non-random) sequence of points $\{\mathbf x_1,\mathbf x_2,\ldots\} \subset \mathbb R^n$ such that
\begin{align}\label{sampling}
\lim_{n\rightarrow \infty}\frac{\#\{j: \mathbf x_j\in B\}}{n} = \lim_{n\rightarrow \infty} \int_{B} \left( \frac{1}{n} \sum_{j=1}^n \delta_{\mathbf x_j}(\mathbf x) \right) \D \mathbf x = \mu(B),
\end{align}
for all rectangles $B = [a_1,b_1] \times \cdots \times [a_n,b_n] \subset \mathbb R^n$.
\end{definition}

\begin{remark}
From a stochastic process point of view, we want the sequence $\{\mathbf x_1,\mathbf x_2,\ldots\}$ to behave like a generic sample path of $\{\mathbf X_1(\omega), \mathbf X_2(\omega),\ldots\}$ where $\mathbf X_i$ are iid vector-valued random variables with joint distribution $\mu$. From the Glivenko--Cantelli theorem \cite[pg.~76]{Durrett2010}, \eqref{sampling} is a necessary consequence of this.  Furthermore, if we can sample with a given method, then we can, in principle, completely determine the distribution $\mu$, even if we do not have an expression for $\mu$.
\end{remark}

If $Y$ samples the uniform distribution on $[0,1]$, then $X$ will sample a distribution that will be  on the order of machine precision ($\approx 2.22\times 10^{-16}$) when compared with the true distribution.  In what follows, we ignore this small error and treat our sampling procedure as exact.

We aim to show that $r_1,\ldots,r_n$  sample the probability measure
\begin{align*}
P_n(x_1,\ldots,x_n) \D \mathbf x = \frac{1}{\hat Z_n} e^{-\sum_{i=1}^n Q(x_i)} \prod_{i<j} (x_i-x_j)^2 \D \mathbf x,
\end{align*}
and we use the representation
\begin{align*}
P_n(x_1,\ldots,x_n) = \frac{1}{n!} \det (\mathcal K_{n}(x_i,x_j))_{1\leq i,j \leq n}
\end{align*}
Instead of sampling an $n$-dimensional distribution, as mentioned above, we reduce the complexity by exploiting the determinantal representation.  Indeed, the sampling is reduced to sampling a sequence of $n$ one-dimensional distributions.  

%
%

We now show that Algorithm~\ref{ChebyshevSample} samples $P_n(x_1,\ldots,x_n)\D\mathbf x$.  As a first step, we show that the above algorithm samples
\begin{align*}
p(x_1,\ldots,x_n)\D \mathbf x = \prod_{j=1}^n \frac{ \vc q_j(x_j)^\top \vc q_j(x_j)}{j}\D \mathbf x.
\end{align*}
This follows from the following lemma.

\begin{lemma}
Assume a probability measure $\mu$ on $ U \subset \mathbb R^n$ has the form
\begin{align*}
\D \mu(x_1,\ldots,x_n) &= F(x_1,\ldots,x_n)\D\mathbf x, \\
&= f_1(x_1)f_2(x_1,x_2)\cdots f_n(x_1,\ldots,x_n)\D\mathbf x, ~~ F > 0,
\end{align*}
where $\int_{\mathbb R^{m-1}}f_m(x_1,\ldots,x_m)\dx_1\cdots \dx_{m}$ defines a probability measure on $\mathbb R$ (for the $x_m$ variable).  If the sequence $\mathbf X_i = (X^i_1,\ldots,X^i_n)$, $i =1,2,\ldots$ samples the standard uniform distribution on $[0,1]^n$ then the sequence, $\mathbf x_i= (x^i_1,\ldots,x^i_n)$, $i=1,2,\ldots$ uniquely defined by the relation
\begin{align}\label{transform}
\begin{split}
X^i_1 &= \int_{-\infty}^{x^i_1} f_1(x_1') \dx_1',\\
X^i_2 &= \int_{-\infty}^{x^i_2} f_2(x^i_1,x_2') \dx_2',\\
\vdots\\
X^i_n &= \int_{-\infty}^{x^i_n} f_n(x^i_1,\ldots,x^i_{n-1},x_n') \dx_n',
\end{split}
\end{align}
samples $\mu$.
\end{lemma}

\begin{proof}
It is clear that the determinant of the Jacobian of the transformation is just $F(x_1^i,\ldots,x_n^i)$ and therefore the mapping $\mathbf x_i \mapsto \mathbf X_i$ is invertible.  We know that for any Borel set $S \subset [0,1]^n$ 
\begin{align*}
\lim_{n\rightarrow \infty}\frac{\#\{j:j=1,\ldots,n ~\text{ and }~ \mathbf X^j \in S\}}{n} = \int_S \D X_1\cdots \D X_n.
\end{align*}
A simple calculation using the determinant of the Jacobian shows that
\begin{align*}
\lim_{n\rightarrow \infty}\frac{\#\{j:j=1,\ldots,n ~\text{ and }~ \mathbf x^j \in T\}}{n} = \int_T F(x_1,\ldots,x_n) \dx_1\cdots \dx_n,
\end{align*}
where $T$ is the inverse image of $S$ under the transformation \eqref{transform}.  Therefore, for any Borel set $T$ let $S$ be the image of $T$ under \eqref{transform} and this shows that the method samples $\mu$.

\end{proof}

Next, we show that $p(x_1,\ldots,x_n) = P_n(x_1,\ldots,x_n)$.   Define
	$$\psi_j^{(k)}(x) = \vc q_k(x_j)^\top \vc q_k(x),$$
so that $\psi_k^{(k)}(x) = \vc f_k^\top \vc q_k(x)$.  
By orthogonality we have
	$$\ip<\psi_j^{(k)},\psi_\ell^{(k)}> = \vc q_k(x_j)^\top \vc q_k(x_\ell).$$
therefore
	\meeq{
		P_n(x_1,\ldots,x_n) = {1 \over n!}  A \qfor \ccr
		 A = \det \begin{mat} \langle \psi_1^{(1)},\psi_1^{(1)} \rangle &  \langle \psi_1^{(1)},\psi_2^{(1)} \rangle & \langle \psi_1^{(1)},\psi_3^{(1)} \rangle & \cdots & \langle \psi_1^{(1)},\psi_n^{(1)} \rangle\\
\langle \psi_2^{(1)},\psi_1^{(1)} \rangle & \langle \psi_2^{(1)},\psi_2^{(1)} \rangle & \cdots &\cdots& \vdots\\
\vdots & \vdots & \ddots & \ddots& \vdots\\
\langle \psi_n^{(1)},\psi_1^{(1)} \rangle & \langle \psi_n^{(1)},\psi_2^{(1)} \rangle & \cdots & \cdots & \langle \psi_n^{(1)}, \psi_n^{(1)} \rangle
\end{mat}.
}

We wish to rewrite $\psi_j^{(1)}$ in terms of $\psi_j^{(k)}(x) = \vc q_k(x_j) \vc q_k(x)$.  First observe that prepending the vector $\vc f_k/\norm{\vc f_k}$ to $Q_k$  we obtain a $k \times k$ orthogonal matrix
	$$\tilde Q_k = \br[{\vc f_k \over \norm{\vc f_k}}, Q_k].$$
  Thus we have  (observing that $\vc f_1 = \vc\phi(x_1)$)
	\meeq{
		\psi_j^{(1)}(x) = \vc\phi(x_j)^\top \tilde Q_1 \tilde Q_1^\top \vc q_1(x) = {\vc \phi(x_j)^\top \vc f_1 \vc f_1^\top \vc q_1(x) \over \norm{\vc f_1}^2} + \vc\phi(x_j)^\top Q_1 Q_1^\top \vc q_1(x) \ccr
			=  c_j \psi_1^{(1)}(x) + \vc q_2(x_j)^\top \vc q_2(x) = c_j \psi_1^{(1)}(x) + \psi_j^{(2)}(x)
		}
By column reductions we get
	$$ A = \det \begin{mat} \langle \psi_1^{(1)},\psi_1^{(1)} \rangle &  \langle \psi_1^{(1)},\psi_2^{(2)} \rangle & \langle \psi_1^{(1)},\psi_3^{(2)} \rangle & \cdots & \langle \psi_1^{(1)},\psi_n^{(2)} \rangle\\
\langle \psi_2^{(1)},\psi_1^{(1)} \rangle & \langle \psi_2^{(1)},\psi_2^{(2)} \rangle & \cdots &\cdots& \vdots\\
\vdots & \vdots & \ddots & \ddots& \vdots\\
\langle \psi_n^{(1)},\psi_1^{(1)} \rangle & \langle \psi_n^{(1)},\psi_2^{(2)} \rangle & \cdots & \cdots & \langle \psi_n^{(1)}, \psi_n^{(2)} \rangle
\end{mat}$$
Similarly, we have
	\meeq{
		\psi_j^{(k)}(x) 
			= c_j \psi_k^{(k)}(x) + \psi_j^{(k+1)}(x), 
		}
and thus we can further reduce to
	$$ A = \det \begin{mat} \langle \psi_1^{(1)},\psi_1^{(1)} \rangle &  \langle \psi_1^{(1)},\psi_2^{(2)} \rangle & \langle \psi_1^{(1)},\psi_3^{(3)} \rangle & \cdots & \langle \psi_1^{(1)},\psi_n^{(n)} \rangle\\
\langle \psi_2^{(1)},\psi_1^{(2)} \rangle & \langle \psi_2^{(1)},\psi_2^{(2)} \rangle & \cdots &\cdots& \vdots\\
\vdots & \vdots & \ddots & \ddots& \vdots\\
\langle \psi_n^{(1)},\psi_1^{(1)} \rangle & \langle \psi_n^{(1)},\psi_2^{(2)} \rangle & \cdots & \cdots & \langle \psi_n^{(1)}, \psi_n^{(n)} \rangle
\end{mat}$$

Noting that for any vectors $\vc u$ and $\vc v$
	$$\ip<\vc u^\top \vc \phi,\vc v^\top \vc \phi> = \vc u^\top \vc v,$$
we have for $j < k$
	\meeq{
		\ip< \psi_j^{(1)}, \psi_k^{(k)}> = \ip<\vc \phi(x_j)^\top \vc \phi, \vc f_k^\top \vc q_k> = \ip<\vc \phi(x_j)^\top \vc \phi, \vc f_k^\top Q_{k-1}^\top \cdots Q_1^\top \vc \phi> \ccr
			= \vc \phi(x_j)^\top Q_1 \cdots Q_{k-1} \vc f_k = \vc q_j(x_j)^\top Q_{k-1}^\top \cdots Q_1^\top  Q_1 \cdots Q_{k-1} \vc f_k \ccr
			= \vc f_j^\top Q_j \cdots Q_{k-1} \vc f_k = 0.
		}
Using this, we subsequently obtain
	\meeq{
		\ip<\psi_k^{(1)},\psi_k^{(k)}> = 		\ip<\psi_k^{(2)},\psi_k^{(k)}> = \cdots = \ip<\psi_k^{(k)},\psi_k^{(k)}>
	}
Thus we have an upper triangular representation:
	\meeq{
	A = \det \begin{mat} \langle \psi_1^{(1)},\psi_1^{(1)} \rangle \\
\vdots & \langle \psi_2^{(2)},\psi_2^{(2)} \rangle & \\
\vdots & \vdots & \ddots  \\
\vdots & \dots  & \hdots  & \langle \psi_n^{(n)}, \psi_n^{(n)} \rangle
\end{mat} \ccr
	= \prod_{k=1}^n \vc q_j(x_j)^\top \vc q_j(x_j).
}
In other words, 
	$$P_n(x_1,\dots,x_n) = {A \over n!} = p(x_1,\dots,x_n).$$

\section*{Acknowledgements}

We would like to thank Paul Cheung, Percy Deift and Govind Menon for discussions concerning this method.  We acknowledge the generous support of the National Science Foundation   through grant NSF-DMS-130318 (TT) and the Australian Research Council through the Discovery Early Career Research Award (SO).  Any opinions, findings, and conclusions or recommendations expressed in this material are those of the authors and do not necessarily reflect the views of the funding sources.

\bibliography{IESampling}

\end{document}